\documentclass[a4paper,UKenglish]{lipics} \input{macros.sty}

\usepackage{microtype}

\title{The Sch\"{u}tzenberger product for Syntactic Spaces\footnote{This project has received funding from the European Research Council (ERC) under the European Union's Horizon 2020 research and innovation programme (grant agreement No.670624).}}

\author{Mai Gehrke} \author{Daniela Petri{\c s}an} \author{Luca
  Reggio}
\affil{IRIF, CNRS and Univ. Paris Diderot, France\\
  \texttt{\{mgehrke,petrisan,reggio\}@liafa.univ-paris-diderot.fr}}
\authorrunning{M.\ Gehrke, D.\ Petri{\c s}an and L.\ Reggio} 

\Copyright{Mai Gehrke, Daniela Petri\c san and Luca Reggio}


\serieslogo{}
\volumeinfo
{}
{0}
{}
{1}
{1}
{1}
\EventShortName{}
\DOI{10.4230/LIPIcs.xxx.yyy.p}

\begin{document}
\maketitle
%
\begin{abstract}
  Starting from Boolean algebras of languages closed under quotients
  and using duality theoretic insights, we derive the notion of
  \emph{Boolean spaces with internal monoids} as recognisers for
  arbitrary formal languages of finite words over finite
  alphabets. This leads to recognisers and syntactic spaces equivalent
  to those proposed in \cite{GGP2010}, albeit in a setting that is
  well-suited for applying existing tools from Stone duality as
  applied in semantics.

  The main focus of the paper is the development of topo-algebraic
  constructions pertinent to the treatment of languages given by logic
  formulas. In particular, using the standard semantic view of
  quantification as projection, we derive a notion of
  \emph{Sch\"{u}tzenberger product} for Boolean spaces with internal
  monoids.  This makes heavy use of the Vietoris construction --- and
  its dual functor --- which is central to the coalgebraic treatment
  of classical modal logic.

  We show that the unary Sch\"{u}tzenberger product for spaces yields
  a recogniser for the language of all models of the formula $\exists
  x.\Phi(x)$, when applied to a recogniser for the language of all
  models of $\Phi(x)$. Further, we generalise global and local
  versions of the theorems of Sch\"{u}tzenberger and Reutenauer
  characterising the languages recognised by the binary
  Sch\"{u}tzenberger product.
  Finally, we provide an equational characterisation of Boolean
  algebras obtained by local Sch\"{u}tzenberger product with the one
  element space based on an Egli-Milner type condition on generalised
  factorisations of ultrafilters on words.
\end{abstract}
\section{Introduction}
This contribution lies at the interface of two distinct areas: One in
semantics concerned with modelling binding of variables, and the other
in the theory of formal languages and the search for separation
results for complexity classes based on a generalisation of the
algebraic theory of regular languages \cite{Straubing1994,KLR2007}.
In semantics of propositional and modal logics, Stone duality and
coalgebraic logic have had great success, but in the presence of
quantifiers more general categorical semantics is required.
Quantifiers change the set of free variables in a formula, leading to
a notion of indexing formulas by their contexts of free variables. In
the theory of regular languages, classes of models indexed by finite
alphabets have long been studied in the form of varieties of
languages~\cite{Eilenberg2}. There, one considers Boolean algebras
of languages closed under quotients over a category of finite
alphabets with monoid morphisms between the corresponding finitely
generated monoids. This paper is intended as a first step towards
establishing a connection between categorical semantics of logics and
fibrational approaches in language theory.

We follow the line set by \cite{GGP2008, GGP2010} and \cite{GKP2016},
which exploits the connection between the algebraic theory of formal
languages and Stone duality, see
also~\cite{BoBoHaPaSi2014,AdamekMUM15}.  In this paper we are
interested in the effect that first-order quantifiers have at the
level of the algebraic recognisers.
%
%
%
This is well understood in the regular case,
where a plethora of powerful tools, in the form of Sch\"{u}tzenberger,
Maltsev, and block products of finite (and profinite) monoids, is
used.
Beyond the regular setting, we take as a departure point classes of
languages equipped with actions of the free monoid over a finite set
and the standard view of existential quantification as projection, and we
derive --- via Stone duality --- our notion of recognisers and of unary
Sch\"{u}tzenberger product.
%
%
Our analysis arrives at an extension of the Sch\"{u}tzenberger 
product, which was originally introduced in~\cite{Schutzenberger65} as a means 
of studying the concatenation product of regular languages and was
further extended in~\cite{Straubing1981} and~\cite{Pin2003} to arbitrary 
arity and to ordered monoids, respectively. 
Reutenauer~\cite{Reutenauer1979}, and Pin~\cite{Pin83} in the ordered 
setting, have provided exact characterisations of the regular languages 
accepted by the Sch\"{u}tzenberger product.

In the setting of regular languages equations have played an essential
r\^ole in providing decidability results for varieties of languages
and various generalisations thereof. For classes of arbitrary
languages decidability is not to be expected and separation of classes
is the main focus. For this reason soundness becomes more important
than completeness per se.  However, complete axiomatisations are
useful for obtaining decidability results for the class of regular
languages within a fragment. See \cite{GKP2016} for an example and for
further motivation relative to the study of circuit complexity
classes.

\textbf{Contributions and Structure.}
 After some preliminaries on Stone duality and actions by monoids, 
 Section~\ref{s:recognition-spaces-dense-monoids} introduces our 
 notion of recognisers and main objects of study, the
 \emph{Boolean spaces with internal monoids}. 
 In Section~\ref{s:unary-schutz-product} we analyse the relation between
 recognisers for a language $L_\Phi$, corresponding to a formula $\Phi$
 with one free first-order variable $x$, and recognisers for the
 existentially quantified language $L_{\exists x.\Phi}$.  To this end,
 in Section~\ref{ss:logical-motivation} we introduce a unary version
 of the Sch\"{u}tzenberger product, $\Ds M$, for a discrete monoid $M$ 
 and prove that if $M$ recognises $L_\Phi$, then $\Ds M$ recognises
 $L_{\exists x.\Phi}$. In Section~\ref{ss:unary-schutz} we extend the 
 unary Sch\"{u}tzenberger product, and the results in 
 Section~\ref{ss:logical-motivation}, to Boolean spaces with internal
 monoids (noting this can be done for semigroups as well). We end the 
 section with a characterisation of the languages recognised by the 
 unary Sch\"{u}tzenberger product $(\Ds X, \Ds S)$ of a Boolean space 
 with an internal semigroup $(X, S)$ (see 
 Theorem~\ref{th:recognised-by-diamond-X}). 
 In Section~\ref{s:binary-schutz-product} we introduce the binary 
 Sch\"{u}tzenberger product of Boolean spaces with internal monoids. 
 Theorems~\ref{th:reutenauer-global} and~\ref{th:reutenauer-local} 
 extend results of Reutenauer in the regular setting and establish the 
 connection with concatenation product for arbitrary languages.
Finally, in Section~\ref{s:ultrafilter-equations} we provide a completeness 
result for the Boolean algebra recognised by the local version of the 
Sch\"utzenberger product of a space with the one element space.\\[-3ex]
\section{Preliminaries}\label{s:preliminaries}
\subsection{Stone duality for Boolean algebras}\label{ss:stone-duality}
Let $(\B,\wedge,\vee,\neg,0,1)$ be a Boolean algebra. Recall that a subset $\UU\subseteq \B$ is a \emph{filter} of $\B$ if it satisfies the following conditions:
\begin{itemize}
\item non-emptiness: $1\in\UU$,
\item upward closure: if $L\in \UU$ and $N\in \B$ satisfies $L\leq N$, then $N\in \UU$,
\item closure under finite meets: if $L,N\in \UU$, then $L\wedge N\in\UU$.
\end{itemize}
A filter $\UU\subseteq\B$ is \emph{proper} if $\UU\neq\B$. \emph{Ultrafilters} are those for which $L\in\UU$ or $\neg L\in\UU$ for each $L\in\B$.
%
In the Boolean algebra $\P(S)$, an example of an ultrafilter is given, for each $s\in S$, by the \emph{principal ultrafilter} associated with the element $s$, 
namely\footnote{Identifying $s\in S$ with $\{s\}\in\P(S)$, we write $\uparrow s$ for $\uparrow \{s\}$.}
\vskip-.7cm
\begin{align}\label{eq:principal-ultrafilter-power-set}
\uparrow s:=\{b\in\P(S)\mid s\in b\}.
\end{align}
\vskip-.2cm
Let $X_{\B}$ be the collection of all the ultrafilters of $\B$. The fundamental insight of Stone is that, equipped with an appropriate topology, one may recover $\B$ from $X_{\B}$. For $L\in\B$ set
\vskip-.7cm
\begin{align}\label{eq:basic-clopens}
\widehat{L}:=\{\UU\in X_{\B}\mid L\in\UU\}.
\end{align}
\vskip-.2cm
Then the family $\{\widehat{L}\mid L\in\B\}$ forms a basis of open sets for a topology $\sigma$ on $X_{\B}$, and the topological space $(X_{\B}, \sigma)$ is called the
 \emph{dual space} of the Boolean algebra $\B$. The topology $\sigma$ is compact, Hausdorff, and admits a basis of \emph{clopen} sets (i.e.\ sets that are 
 both open and closed) since the complement of $\widehat{L}$ is $\widehat{\neg L}$.
Compact Hausdorff spaces that admit a basis of clopen sets are known as \emph{Boolean} (or \emph{Stone}) \emph{spaces}. The collection of clopens of a Boolean space $X$ (equipped with set-theoretic operations) constitutes a Boolean algebra, known as the \emph{dual algebra} of $X$. 
These processes are, up to natural equivalence, inverse to each other.
Given a morphism of Boolean algebras $h\colon\A\to\B$, the inverse image map on their power sets $h^{-1}\colon\P(\B)\to\P(\A)$ sends ultrafilters to ultrafilters and 
provides the continuous map from the dual space of $\B$ to the dual space of $\A$. Similarly, the inverse image map of a continuous map $f\colon X\to Y$ provides the morphism from the dual algebra of $Y$ to that of $X$. 
In this correspondence, quotient algebras correspond to embeddings as (closed) subspaces, and inclusions as subalgebras correspond to quotient spaces. 
%
In category-theoretic terms, this establishes a contravariant equivalence between the category of Boolean spaces and continuous maps, and the category of Boolean algebras and their morphisms. This
is the content of the celebrated Stone duality for Boolean algebras \cite[Theorems 67 and 68]{Stone1936}.
%
%

We end this section with an example of a Boolean algebra and its dual space which will play a key r\^{o}le in the sequel. Let $S$ be a set. Then $\P(S)$ is a 
Boolean algebra and its dual space, denoted by $\beta(S)$, is known as the \emph{Stone-\v{C}ech compactification} of the set $S$.
We remark that the map $\iota\colon S\to \beta(S)$, mapping an element $s$ to the principal ultrafilter $\uparrow s$ of~\eqref{eq:principal-ultrafilter-power-set}, 
is injective and embeds $S$, with the discrete topology, as a dense subspace of $\beta(S)$.
%
%
Henceforth, we will consider $S$ as a subspace of $\beta(S)$, identifying $s\in S$ with $\uparrow s$, thus suppressing the embedding $\iota$.
The space $\beta(S)$ is characterised by the following \emph{universal property}: if $X$ is a compact Hausdorff space and $f\colon S\to X$ is any function, 
then there is a (unique) continuous function $g\colon \beta(S)\to X$ such that the following diagram commutes.
\begin{equation}\label{eq:stone-cech-universal-property}
\begin{tikzcd}
S \arrow[hookrightarrow]{r}{} \arrow{dr}[swap]{f} & \beta(S) \arrow{d}{g} \\
 & X
\end{tikzcd}
\end{equation}
Consequently, if $T$ is a discrete space, any function $f\colon S\to T$ can be extended to a continuous map $\beta(f)\colon\beta(S)\to\beta(T)$. 
Explicitly, the latter is given, for each $\UU\in\beta(S)$ and $L\in\P(T)$, by
\begin{align}\label{eq:beta-on-maps}
L\in\beta(f)(\UU) \quad \text{if, and only if,} \quad f^{-1}(L)\in\UU.
\end{align}
%
%
%
\subsection{Monoid actions}\label{ss:monoid-actions}
Let $(M,\cdot,1)$ be a monoid, and $X$ be a set. A function $\lambda\colon M\times X\to X$ is called a \emph{left action} of $M$ on $X$ provided
\begin{itemize}
\item for all $x\in X$, $\lambda(1,x)=x$,
\item for all $m,m'\in M$ and $x\in X$, $\lambda(m\cdot m',x)=\lambda(m,\lambda(m',x))$.
\end{itemize}
\vskip.1cm
Similarly, one can define a \emph{right action} $\rho\colon X\times M\to X$ of $M$ on $X$.
For each $m\in M$, we refer to the function $\lambda_m\colon X\to X$ given by $\lambda_m(x):=\lambda(m,x)$ (respectively to the function $\rho_m\colon X\to X$ given by $\rho_m(x):=\rho(x,m)$) as the \emph{component} of the action $\lambda$ at $m$ (respectively, of the action $\rho$ at $m$). A pair consisting of left and right 
actions $\lambda,\rho$ of $M$ on $X$ is said to be \emph{compatible} if, for all $m,m'\in M$, $\lambda_{m}\circ\rho_{m'}=\rho_{m'}\circ\lambda_{m}$. 
We call such a pair of compatible actions  a \emph{biaction} of $M$ on $X$ (or an \emph{$M$-biaction} on $X$).
\vspace*{-.2cm}
\begin{example}\label{ex:monoid-acting-on-itself}
Any monoid $M$ can be seen as acting on itself on the left and on the right. The component of the left action at $m\in M$ is the multiplication on the left by $m$, and the 
component of the right action is the multiplication on the right by $m$. The compatibility of the two actions amounts precisely to the associativity of the monoid operation.
\end{example}
\vspace*{-.4cm}
%
%
\begin{example}\label{ex:N}
Consider $\nbb$, the free monoid on one generator. As observed in Example \ref{ex:monoid-acting-on-itself}, for each $n\in\nbb$ we have components
$\lambda_n,\rho_n\colon \nbb\to\nbb$ of compatible left and right actions of $\nbb$ on itself.
By the universal property~\eqref{eq:stone-cech-universal-property} of the Stone-\v{C}ech compactification, we obtain continuous components 
$\beta(\lambda_n),\beta(\rho_n)\colon\beta(\nbb)\to\beta(\nbb)$ of a biaction of $\nbb$ on $\beta(\nbb)$.
However the set $\beta(\nbb)$ is not equipped with a continuous monoid operation, see \cite[Chapter 4]{HS2012}.
\end{example}
\section{Recognition by spaces with dense monoids}\label{s:recognition-spaces-dense-monoids}
We start by showing how our main objects of study (see Definition
\ref{d:spaces-with-internal-monoids} below) arise naturally by
considering duals of Boolean algebras of languages closed under
certain operations known as quotients by words.
%
Let $\Alp$ be a finite alphabet. Instantiating the monoid in
Example~\ref{ex:monoid-acting-on-itself} with the free monoid $\Alp^*$
on $\Alp$, we obtain a biaction of $\Alp^*$ on itself.
The components of the left and right actions are given by
concatenation, and they will be denoted by
\vskip-.7cm
\[
  \lambda_w\colon \Alp^*\to\Alp^*, \, \, u\mapsto wu \quad
  \text{and} \quad \rho_w\colon \Alp^*\to\Alp^*, \, \,
  u\mapsto uw.
\]
\vskip-.2cm
These actions can be dualised from $\Alp^*$ to
$\P(\Alp^*)$. The right $\Alp^*$-action on $\P(\Alp^*)$ is given by
$\lambda_w^{-1}\colon\P(\Alp^*)\to\P(\Alp^*)$, while the left action
is given by $\rho_w^{-1}\colon\P(\Alp^*)\to\P(\Alp^*)$. These are the
well-known \emph{left quotients} and \emph{right quotients} of
language theory given, respectively, by
\vskip-.7cm
\begin{align*}
  L\mapsto \{u \mid wu\in L\}=:w^{-1}L \quad \text{and}\quad  L\mapsto \{u \mid uw\in L\}=:Lw^{-1}.
\end{align*}
\vskip-.2cm
It is immediate that the $\lambda_w^{-1}$ and $\rho_w^{-1}$ are
homomorphisms and compatible $\Alp^*$-actions. 

Dualising again, we see that the space $\beta(\Alp^*)$ is equipped
with (compatible and continuous) left and right $\Alp^*$-actions
given, for all $w\in\Alp^*$, by $\beta(\lambda_w)$ and
$\beta(\rho_w)$, respectively. By abuse of notation and for ease of
readability, we will denote these actions again by $\lambda_w$,
respectively $\rho_w$. We notice that the pair $(\beta(\Alp^*),\Alp^*)$ 
exhibits the following structure:
\vspace*{.1cm}
  \begin{itemize}
  \item a Boolean space $\beta(\Alp^*)$,
  \item a dense subspace $\Alp^*$ equipped with a monoid structure,
  \item a biaction of $\Alp^*$ on $\beta(\Alp^*)$ with continuous components
  extending that of $\Alp^*$ on itself. 
\end{itemize}

Now, consider a Boolean subalgebra $\B$ of $\P(\Alp^*)$ closed 
under left and right quotients by words. Then the maps 
$\lambda_w^{-1}$ and $\rho_w^{-1}$ restrict to Boolean algebra 
morphisms on $\B$, yielding the following commutative diagrams.
\begin{equation}
  \label{eq:actions-on-boolean-algebra}
  \begin{tikzcd} 
    \P(\Alp^*)\arrow{r}{\lambda_w^{-1}} & \P(\Alp^*) & & \P(\Alp^*)\arrow{r}{\rho_w^{-1}}& \P(\Alp^*)  \\
    \B\arrow[hookrightarrow]{u}\arrow[dashed]{r}{\lambda_w^{-1}} &
    \B\arrow[hookrightarrow]{u} & & \B\arrow[dashed]{r}{\rho_w^{-1}}
    \arrow[hookrightarrow]{u} & \B\arrow[hookrightarrow]{u}
  \end{tikzcd}
\end{equation}
\vskip-.2cm
\noindent Let $X_\B$ denote the dual space of the Boolean algebra $\B$. The
embedding $\B\hookrightarrow\P(\Alp^*)$ dually corresponds to a quotient 
 $\tau\colon\beta(\Alp^*)\epi X_\B$.
%
%
%
The space $X_\B$ also admits left and right $\Alp^*$-actions induced
by the duals of the maps $\lambda_w^{-1}$, respectively $\rho_w^{-1}$,
from~\eqref{eq:actions-on-boolean-algebra}. We thus obtain
%
\vskip-.5cm
\begin{equation}
  \label{eq:actions-on-dual-space}
  \begin{tikzcd}
    \beta(\Alp^*)\arrow{r}{\lambda_w}\arrow[twoheadrightarrow]{d}[swap]{\tau}
    & \beta(\Alp^*)\arrow[twoheadrightarrow]{d}{\tau} & &
    \beta(\Alp^*)\arrow{r}{\rho_w}
    \arrow[twoheadrightarrow]{d}[swap]{\tau} & \beta(\Alp^*)\arrow[twoheadrightarrow]{d}{\tau} \\
    X_\B\arrow[dashed]{r}{\lambda_w} & X_\B & & X_\B\arrow[dashed]{r}{\rho_w}& X_\B
  \end{tikzcd}
\end{equation}
%
  Then $M:=\tau[\Alp^*]$ is a dense subspace of $X_{\B}$, and we have the
  following commutative diagram.
  \vskip-.5cm
  \begin{equation}
    \label{eq:dense-monoid-in-syntactic-space}
    \begin{tikzcd} 
      \beta(\Alp^*) \arrow[twoheadrightarrow]{r}{\tau} & X_{\B} \\
      \Alp^* \arrow[twoheadrightarrow]{r}{\tau}
      \arrow[hookrightarrow]{u}{} & M
      \arrow[hookrightarrow]{u}[swap]{}
    \end{tikzcd}
  \end{equation}
%
  We observe that the pair $(X_{\B},M)$ exhibits the same kind of
  structure as $(\beta(\Alp^*),\Alp^*)$:
  \begin{itemize}
  \item a Boolean space $X_{\B}$,
  \item a dense subspace $M$ equipped with a monoid structure,
  \item a biaction of $M$ on $X_{\B}$ with continuous components
  extending the biaction of $M$ on itself. 
  \end{itemize}
  \vskip.1cm
\noindent  Indeed, recall that $X_{\B}$ is equipped with left and right
  $\Alp^*$-actions which are preserved by the map $\tau$ by
  commutativity of~\eqref{eq:actions-on-dual-space}.
 The $\Alp^*$-actions on $X_{\B}$ restrict to $\Alp^*$-actions on $M$, 
 which are preserved by the restriction of $\tau$.
  The monoid structure on $M$ is then defined as follows. For any
  $m\in M$ pick $w_{m}\in\Alp^*$ satisfying $\tau(w_{m})=m$. Such an
  element exists because $M$ is the image of $\Alp^*$ by $\tau$. For $m,m'\in M$,
  set $m\cdot m':=\lambda_{w_m}(m').$
  It is easily seen that the latter operation is well-defined and
  provides a monoid structure on $M$ which makes the restriction of $\tau$ a monoid
  morphism.

  As first introduced in \cite{GGP2010}, we will be using dual spaces
  equipped with actions as recognisers. The examples above motivate
  the following definition.

\begin{definition} \label{d:spaces-with-internal-monoids} A
  \emph{Boolean space with an internal monoid} is a pair $(X,M)$
  consisting of
  \begin{itemize}
  \item a Boolean space $X$,
  \item a dense subspace $M$ equipped with a monoid structure,
  \item a biaction of $M$ on $X$ with continuous components
  extending the biaction of $M$ on itself. 
  \end{itemize}
\end{definition}
%

%

\begin{remark}
\label{remark-pervin} 
The recognisers introduced in \cite{GGP2010} are 
monoids equipped with a uniform space structure, namely the Pervin 
uniformity given by a Boolean algebra of subsets of the monoid, so that 
the biaction of the monoid on itself has uniformly continuous components.
Such an object was called a \emph{semiuniform monoid}. One may show 
that the completion of a semiuniform monoid is a Boolean space with an
internal monoid. Conversely, given a Boolean space with an internal monoid
$(X,M)$, the Pervin uniformity on $M$ induced by the dual of $X$ is a 
semiuniform monoid, and these two constructions are inverse to each other.

%
\end{remark}
%
%
We are interested in maps between pairs $(X,M)$ and $(Y,N)$, i.e.\
continuous maps $X\to Y$ which preserve the additional structure.
%
\begin{definition} \label{d:morphisms-of-spaces-with-internal-monoids}
  A \emph{morphism} between two Boolean spaces with internal monoids
  $(X,M)$ and $(Y,N)$ is a continuous map $f\colon X\to Y$ such that
  $f$ restricts to a monoid morphism $M\to N$.
\end{definition}
Morphisms, as just defined, are in fact also biaction-preserving maps.
%
\begin{lemma}\label{l:morphism-are-action-preserving}
  Let $f\colon(X,M)\to(Y,N)$ be a morphism of Boolean spaces with
  internal monoids. Then $f$ preserves the actions, i.e.\
for every $m\in M$
\vspace*{-.2cm}
  \begin{align*}
    f\circ \lambda_{m}=\lambda_{f(m)}\circ f \quad\text{and}\quad  f\circ \rho_{m}=\rho_{f(m)}\circ f.
  \end{align*}
\end{lemma}
\begin{example}\label{ex:dense-monoid-in-X}
  The map $\tau\colon(\beta(\Alp^*),\Alp^*)\to (X_{\B},M)$
  of~\eqref{eq:dense-monoid-in-syntactic-space} is a morphism of
  Boolean spaces with internal monoids.
\end{example}

\begin{remark}
The map $L\mapsto\wL$ of~\eqref{eq:basic-clopens} establishes a one-to-one
correspondence between the elements of $\P(\Alp^*)$ and the 
clopens of $\beta(\Alp^*)$. 
Thus, we will sometimes blur the distinction between recognition of a
language $L$ and recognition of the corresponding clopen $\wL$.
\end{remark}

\begin{definition}
 \label{def:recognition}
  Let $\Alp$ be a finite alphabet, and let $L\subseteq\P(\Alp^*)$
  be a language.
  We say that $L$ (or $\wL$) is \emph{recognised by the morphism}
  $f\colon(\beta(\Alp^*),\Alp^*)\to(X,M)$ if there is a clopen
  $C\subseteq X$ such that $\wL=f^{-1}(C)$.
  Moreover, the language $L$ is \emph{recognised by the
    space} $(X,M)$ if there is a
  morphism $(\beta(\Alp^*),\Alp^*)\to(X,M)$ recognising $L$.
 Similarly, we say that a morphism (or a space) recognises a 
 Boolean algebra if it recognises all its elements. 
\end{definition}
\begin{remark}
  In general, a morphism $(\beta(\Alp^*),\Alp^*)\to (X,M)$ with
  \emph{infinite} $M$, recognises (in the sense of
  Definition~\ref{def:recognition}) far less languages than the
  induced monoid morphism $\Alp^*\to M$.  On the other hand, a finite
  monoid $M$ may be seen as a space with an internal monoid, in which
  the space component is the monoid itself, equipped with the
  discrete topology. A morphism $(\beta(\Alp^*),\Alp^*) \to(M,M)$
  yields in particular a monoid morphism $\Alp^*\to M$.  Conversely, a
  monoid morphism $h\colon\Alp^*\to M$ extends uniquely to a
  continuous map $\beta h\colon\beta(\Alp^*)\to M$ whose restriction
  to $\Alp^*$ is a monoid morphism. Thus the notion of recognition
  introduced here extends the usual notion for regular languages, but
  is finer-grained in the non-regular setting.
\end{remark}

\section{A unary variant of the Sch\"{u}tzenberger product}\label{s:unary-schutz-product}

\subsection{Logical motivation: existentially quantified languages}
\label{ss:logical-motivation}

Consider the free monoid $\Alp^*$ over a finite alphabet $\Alp$.
A word $w\in\Alp^*$ may be seen as a structure based on the set
$\{0,\ldots,|w|-1\}$,\footnote{Here, as usual, $|w|\in\nbb$ denotes
  the length of the word $w=w_0\cdots w_{|w|-1}\in\Alp^*$.}  equipped
minimally with a unary predicate for each letter $a\in\Alp$, which
holds at $i$ if and only if $w_i=a$. Now given a formula $\Phi$ (in a
language interpretable over words as structures), assumed for
simplicity to have only one free first-order variable $x$, we will see
the set $L_\Phi$ of all words satisfying $\Phi$ as a language in the
extended alphabet $\Atw$.
In the terminology of~\cite{Straubing1994}, 
$L_\Phi$ consists of $\{x\}$-structures, which correspond to words in the 
subset $(\Sigma\times\{0\})^*(\Sigma\times\{1\})(\Sigma\times\{0\})^*$ of 
the free monoid $(\Atw)^*$. An $\{x\}$-structure satisfies $\Phi$ provided 
the underlying word in the alphabet $\Alp$ satisfies $\Phi$ under the 
interpretation in which $x$ points to the unique position marked with a $1$. 
Notice that $(\Sigma\times\{0\})^*(\Sigma\times\{1\})(\Sigma\times\{0\})^*$ is
isomorphic to the set $\Alp^*\otimes\nbb$ of words in $\Alp^*$ with a
marked spot defined by
\begin{align*}
\Alp^*\otimes\nbb:=\{(w,i)\in \Alp^*\times\nbb\mid i<|w|\}.
\end{align*}
Throughout this section we will make use of the following three maps
\begin{align*}
&  \gamma_0\colon\Alp^*\to(\Atw)^*, & 
&  \gamma_1\colon\Alp^*\otimes\nbb\to(\Atw)^*, &
&  \pi\colon\Alp^*\otimes\nbb\to\Alp^*.
\end{align*}
\begin{itemize}
\item The map $\gamma_0\colon\Alp^*\to (\Atw)^*$ is the embedding given
by $w\mapsto w^0$, where $w^0$ has the same length as $w$ and
\[
(w^0)_j:=(w_j,0)\quad \text{for each} \quad j<|w|.
\]
\item The map $\gamma_1\colon\Alp^*\otimes\nbb\to (\Atw)^*$ is the embedding
given by
$
(w,i)\mapsto w^{(i)},
$
where $w^{(i)}$ has the same length as $w$ and 
\[
(w^{(i)})_j:=\left\{\begin{array}{rcl}
    (w_j,0)&\text{if}&i\neq j<|w|\\
    (w_i,1)&\text{if}&i=j.
  \end{array}
\right.
\]
\item The map $\pi\colon\Alp^*\otimes\nbb\to \Alp^*$ is the projection on the
first coordinate.
\end{itemize}

\begin{remark} 
  The language $L_{\exists x.\Phi}$ is obtained as
  $\pi[\gamma_1^{-1}(L_\Phi)]$. More generally, given a language
  $L\subseteq(\Atw)^*$, we shall denote $\pi[\gamma_1^{-1}(L)]\subseteq \Alp^*$ by $L_\exists$.
\end{remark}

\begin{remark}
Notice that, unlike $\gamma_0$, the maps $\gamma_1$ and $\pi$ are not
monoid morphisms. Indeed, $\Alp^*\otimes\nbb$ does not have a suitable
monoid structure. However, $\Alp^*\otimes\nbb$ does have a
$\Alp^*$-biaction structure. For $v\in\Alp^*$, the components of the
left and right actions are given by
\begin{align*}
& \lambda_v(w,i):=(vw,i+|v|),\\
& \rho_v(w,i):=(wv,i).
\end{align*}
It is clear that both $\gamma_1$ and $\pi$ preserve the
$\Alp^*$-actions.
\end{remark}

Assume that the language $L_\Phi$ is recognised by a monoid morphism
$\tau\colon(\Atw)^*\to M$.  We have the following pair of
functions\footnote{Notice that this is not a relational morphism in
  the sense of Tilson's definition given in~\cite{Eilenberg2}, since the domain $\Alp^*\otimes\nbb$ does not
  have a compatible monoid structure.} with domain $\Alp^*\otimes\nbb$
\begin{equation*}\label{eq:span}\begin{tikzcd}[row sep=0.5em,column sep=3em]
 & \Alp^*\otimes\nbb \arrow{dl}[swap]{\pi} \arrow{dr}{\gamma_1}& & \\
 \Alp^* & & (\Atw)^* \arrow{dr}{\tau} & \\
 & & & M
\end{tikzcd}\end{equation*}
which gives rise to a relation $R\colon\Alp^*\nrightarrow M$ given by
\[
(w,m)\in R \quad \text{if, and only if,}\quad \exists
(w,i)\in\pi^{-1}(w).\ (\tau\circ\gamma_1)(w,i)=m.
\]
Though $\pi$ is not injective, it does have \emph{finite preimages}. 
As will be crucial in what follows, this allows us to represent $R$ as a 
function (which, in general, is not a monoid morphism)
\begin{align}
  \label{eq:rel-PM}
  \xi_1\colon\Alp^* \to\Pfin(M),\quad w\mapsto \{\tau(w^{(i)})\mid 0\le i< |w|\}
\end{align}
where $\Pfin(M)$ denotes the set of finite subsets of $M$.  Consider
the monoid structure on $\Pfin(M)$ with union as the multiplication,
and the empty set as unit. Notice that the monoid $M$
acts on $\Pfin(M)$ both to the left and to the right, and the two
actions are compatible. The left action $M\times\Pfin(M)\to \Pfin(M)$
is given, for $m\in M$ and $S\in\Pfin(M)$, by $m\cdot S:=\{m\cdot s\mid s\in
S\}$. Similarly, the right action is given by $S\cdot m:=\{s\cdot m\mid s\in
S\}$.
\begin{definition}\label{def:unary-schutz-monoid}
  We define the \emph{unary Sch\"{u}tzenberger product} $\Ds M$ of $M$ as the
  bilateral semidirect product $\Pfin(M)*M$ of the monoids
  $(\Pfin(M),\cup)$ and $(M,\cdot)$. Explicitly, the underlying set of
  this monoid is the Cartesian product $\Pfin(M)\times M$, and the
  multiplication $*$ on $\Pfin(M)*M$ is given by
\begin{equation*}
  \label{eq:1}
  (S,m)*(T,n):=(S\cdot n\cup m\cdot T,m\cdot n).
\end{equation*}
\end{definition}
Note that the projection onto the second coordinate, 
$\pi_2\colon\Ds M\to M$, is a monoid morphism.

\begin{proposition}
  \label{prop:monoid-blprod}
  If $\tau\colon(\Atw)^*\to M$ is a monoid morphism recognising $L_\Phi$,
  then there exists a monoid morphism
\begin{equation*}
  \label{eq:rec-mon}
  \xi\colon\Alp^*\to \Ds M
\end{equation*}
that recognises the language
$L_{\exists x.\Phi}$ and makes the following diagram commute.
\begin{equation*} 
  \label{eq:commdiag} 
  \begin{tikzcd}
    \Alp^* \arrow{r}{\xi} \arrow{d}[swap]{\gamma_0} & \Ds M \arrow{d}{\pi_2} \\
    (\Atw)^* \arrow{r}{\tau} & M
  \end{tikzcd} 
\end{equation*}
\end{proposition}
\begin{proof}[Proof idea]
   The map $\xi$ is obtained by pairing
  $\xi_1\colon\Alp^*\to\Pfin(M)$ of~\eqref{eq:rel-PM} and $\tau\circ
  \gamma_0\colon\Alp^*\to M$. Explicitly,
\[
w\mapsto (\{\tau(w^{(i)})\mid 0\le i< |w|\},\tau(w^0)).
\]
%
One may show that the map $\xi$ is a monoid
morphism with respect to the concatenation on $\Alp^*$ and the
multiplication $*$ on the semidirect product $\Pfin(M)*M$. 
%
Now let $V$ be a subset of $M$ such that $L_\Phi=\tau^{-1}(V)$, and
consider the set $\Dv V\subseteq\Pfin(M)$ defined as
$\{S\in\Pfin(M)\mid S\cap V\neq\emptyset\}$. Then $\xi^{-1}(\Dv V\times
M)$ is precisely $L_{\exists x.\Phi}$.
\end{proof} 

\begin{remark} 
In~\cite{Straubing1981} Straubing generalised the Sch\"{u}tzenberger product for any finite
number of monoids. Using his construction, the unary
Sch\"{u}tzenberger product of $M$ is simply $M$, and hence is
different from $\Ds M$ introduced above. 

For the connection between closure under concatenation product and first-order quantification in the regular setting, see \cite{MP1971}.
\end{remark}
\begin{remark} 
For lack of space, we have chosen to just `pull Definition~\ref{def:unary-schutz-monoid}
(and consequently also the upcoming Definition~\ref{def:unary-schutz-spaces}) out of a hat'. 
However, by a careful analysis of how quotients in $\P(\Alp^*)$ of languages $L_\exists$ 
are calculated, relative to corresponding calculations in $\P((\Alp\times 2)^*)$, one may simply 
derive by duality that the operation given here is the right one.
\end{remark}
\subsection{The Sch\"{u}tzenberger product for one space $\Ds X$}
\label{ss:unary-schutz}

In this section we assume that the language $L_\Phi\subseteq(\Atw)^*$ is
recognised by a morphism of Boolean spaces with internal monoids
$ \tau\colon (\beta(\Atw)^*,(\Atw)^*)\to(X,M)$.
Notice that in this case we have a pair of continuous maps
\begin{equation}\label{eq:topo-span}\begin{tikzcd}[row sep=0.5em,column sep=3em]
 & \beta(\Alp^*\otimes\nbb) \arrow{dl}[swap]{\beta\pi} \arrow{dr}{\beta\gamma_1}& & \\
 \beta(\Alp^*) & & \beta(\Atw)^* \arrow{dr}{\tau} & \\
 & & & X
\end{tikzcd}\end{equation}
which, as before, yields a relation $\beta(\Alp^*)\nrightarrow X$. We
would like to describe this relation as a continuous map on
$\beta(\Alp^*)$. To this end, we need an analogue for spaces of the 
finite power set construction. This is provided by the \emph{Vietoris
space construction} (see Section \ref{ss:Vietoris} in the appendix for further details).

\begin{definition}
  Let $X$ be a Boolean space. The \emph{Vietoris space} $\V(X)$ is the Boolean space
  with underlying set $\{K\subseteq X\mid K\ \text{is closed in}\ X\}$, and topology
  generated by the subbasis consisting of the sets, for $V$ clopen in $X$, of the
  form 
  \vspace{-.2cm}
  \begin{align*}
  \boxa V:=\{K\in\V(X)\mid K\subseteq V\} \quad \text{and} \quad \Dv V:=\{K\in\V(X)\mid K\cap V\neq\emptyset\}.
  \end{align*}
  
  %
\end{definition}
\vspace{-.2cm}
Just as in the monoid case, diagram~\eqref{eq:topo-span} yields a map
\vspace{-.1cm}
  \begin{align}\label{eq:map-xi-1-for-spaces}
  \xi_1\colon\beta(\Alp^*)\to \V(X)
  \end{align}
  \vspace{-.5cm}
   
 \noindent defined as the composition $\tau\circ\beta\gamma_1\circ (\beta\pi)^{-1}$, or equivalently 
 as the unique continuous extension of the map $\xi_1\colon\Alp^*\to\Pfin(M)$ defined in~\eqref{eq:rel-PM}.
 %
%
\begin{definition}
\label{def:unary-schutz-spaces}
  We define the \emph{unary Sch\"{u}tzenberger product} of a Boolean
  space with an internal monoid $(X,M)$ as the pair $(\Ds X,\Ds M)$,
  where $\Ds X$ is the space $\V(X)\times X$ equipped with the product
  topology and $\Ds M$ is as in
  Definition~\ref{def:unary-schutz-monoid}.
\end{definition}

\begin{lemma}\label{l:unary-schutz-product-is-a-bin}
  The unary Sch\"{u}tzenberger product $(\Ds X,\Ds M)$ of $(X,M)$ is a Boolean space with an
  internal monoid.
\end{lemma}
\begin{proof}[Proof Idea]
Recall that $M$ is a dense subspace of $X$. It follows by
  Lemma~\ref{l:finite-powerset-dense-in-vietoris} in Appendix~\ref{a:appendix-unary-schutz-product}
  that $\Pfin(M)$ is a dense subspace of $\V(X)$. Thus the monoid $\Ds M$ is a dense subspace of
  $\Ds X$. Next we define the actions of $\Ds M$ on $\Ds X$ as
  follows:
  \vspace{-.2cm}
  \begin{align*}
  \begin{split}
    l_{(S,m)}(T,x):=(\{\lambda_s(x)\mid s\in S\}\cup\lambda_m[T],\lambda_m(x)),\\
    r_{(S,m)}(T,x):=(\{\rho_s(x)\mid s\in S\}\cup\rho_m[T],\rho_m(x)).
  \end{split}
  \end{align*}
  
  \vspace{-.2cm}
  
 \noindent It is not difficult to see that the above maps are the
  unique continuous extensions to $\Ds X$ of the multiplication by $(S,m)$, to the
  left and to the right, on $\Ds M$.
\end{proof}  


The projection  $\pi_2\colon\Ds X\to X$
is a morphism of Boolean spaces with internal monoids.

\begin{proposition}
\label{prop:recog-ex-unary-schutz}
  If $\tau\colon (\beta(\Atw)^*,(\Atw)^*)\to(X,M)$ is a morphism of
  Boolean spaces with internal monoids recognising $L_\Phi$, then there 
  is a morphism 
   $ \xi\colon(\beta(\Alp^*),\Alp^*)\to(\Ds X,\Ds M)$
recognising $L_{\exists x.\Phi}$ and such that the following diagram commutes.
\vspace*{-.2cm}
\begin{equation*} 
  \label{eq:rec-un-schutz=sp} 
  \begin{tikzcd}
    \beta(\Alp^*) \arrow{r}{\xi} \arrow{d}[swap]{\beta\gamma_0} & \Ds X \arrow{d}{\pi_2} \\
    \beta(\Atw)^* \arrow{r}{\tau} & X
  \end{tikzcd} 
\end{equation*}
\vspace{-.2cm}
\end{proposition}
\vspace{-.2cm}
All the constructions introduced so far can be carried out for semigroups. In particular, we can consider
Boolean spaces with internal semigroups as recognisers of languages in $\P(\Alp^+)$.  Along the lines of
Definition~\ref{def:unary-schutz-monoid}, we introduce the unary
Sch\"{u}tzenberger product $\Ds S$ of a semigroup $S$ as the bilateral
semidirect product of the semigroups $(\Pfin^+(S),\cup)$ and
$(S,\cdot)$, where $\Pfin^{+}(S)$ denotes the family of finite non-empty
subsets of $S$. Similarly, at the level of spaces, in the Vietoris construction we
will consider only non-empty closed subsets.

Now, write $\B(X,\Alp)$ for the Boolean algebra of languages in
$\P(\Alp^{+})$ recognised by the Boolean space with an internal semigroup
$(X,S)$, and note that the latter Boolean algebra is always closed under
quotients. Moreover, given a language
  $L\subseteq(\Atw)^+$, recall that $L_\exists$ denotes the language $\pi[\gamma_1^{-1}(L)]$.
\begin{theorem}
  \label{th:recognised-by-diamond-X}
  Let $(X,S)$ be a Boolean space with an internal semigroup, and let 
  $\B(X,\Atw)_\exists$ denote the Boolean subalgebra closed under
  quotients of $\P(\Alp^+)$ generated by the family
  $\{L_\exists\mid L\in\B(X,\Atw)\}$. Then $\B(\Ds X,\Alp)$ coincides with the Boolean algebra
  generated by the union of $\B(X,\Alp)$ and $\B(X,\Atw)_\exists$.  
\end{theorem}  

The proof of this theorem 
hinges on
the fact that the first components of the recognising morphisms
evaluate to non-empty subsets. An analogous statement can be
formulated for monoids, but we would have to restrict the recognising
morphisms when defining $\B(\Ds X,\Alp)$.
\section{A variant of the Sch\"{u}tzenberger product for two spaces}\label{s:binary-schutz-product}
Given two monoids $(M,\cdot),(N,\cdot)$, the Sch\"{u}tzenberger product $\Ds(M,N)$ can be defined as the monoid $\Pfin(M\times N)\times M\times N$ whose operation
is given by
\begin{align*}
(S,m_1,n_1)\cdot (T,m_2,n_2):=(m_1\cdot T\cup S\cdot n_2, m_1\cdot m_2, n_1\cdot n_2).
\end{align*}
Now, consider two Boolean spaces with internal monoids $(X,M)$ and
$(Y,N)$. We define the space $\Ds(X,Y)$ as the product $\V(X\times
Y)\times X\times Y$. It is clear that the monoid $\Ds(M,N)$ is dense in $\Ds(X,Y)$.
Moreover, the left action of $\Ds(M,N)$ on itself can be extended to $\Ds(X,Y)$ by setting, for any $(S,m_1,n_1)\in \Ds(M,N)$,
\begin{align}\label{eq:left-action-binary-schutz}
\lambda_{(S,m_1,n_1)}\colon \Ds(X,Y)\to \Ds(X,Y), \ (Z,x,y)\mapsto (m_1Z\cup Sy, \lambda_{m_1}(x), \lambda_{n_1}(y)),
\end{align}
where
\begin{align*}
\hskip-5pt m_1 Z:=\{(\lambda_{m_1}(x),y)\in X\times Y\mid (x,y)\in Z\} \ \text{ and }\ S y:=\{(m,\lambda_{n}(y))\in X\times Y\mid (m,n)\in S\}.
\end{align*}
Similarly, the right action can be defined by
\begin{align}\label{eq:right-action-binary-schutz}
\rho_{(S,m_1,n_1)}\colon \Ds(X,Y)\to \Ds(X,Y), \ (Z,x,y)\mapsto (Zn_1\cup xS, \rho_{m_1}(x), \rho_{n_1}(y)),
\end{align}
where
\begin{align*}
\hskip-5pt Zn_1:=\{(x,\rho_{n_1}(y))\in X\times Y\mid (x,y)\in Z\} \ \text{ and }\ xS:=\{(\rho_{m}(x),n)\in X\times Y\mid (m,n)\in S\}.
\end{align*}
It is easy to see that we obtain a biaction of $\Ds(M,N)$ on $\Ds(X,Y)$. Furthermore,
\begin{lemma}\label{l:action-on-binary-schutz-continuous}
The biaction of $\Ds(M,N)$ on $\Ds(X,Y)$ defined in~\eqref{eq:left-action-binary-schutz} and~\eqref{eq:right-action-binary-schutz} has continuous components. 
Thus $(\Ds(X,Y),\Ds(M,N))$ is a Boolean space with an internal monoid.
\end{lemma}
%

%

The next three results establish the
connection between concatenation of possibly non-regular languages and
the Sch\"{u}tzenberger product of Boolean spaces with internal
monoids. We thus extend the theorems of Sch\"{u}tzenberger~\cite{Schutzenberger65} and
Reutenauer~\cite{Reutenauer1979}.

\begin{theorem}[Reutenauer's theorem, global version]
 \label{th:reutenauer-global}
  Consider Boolean spaces with dense monoids $(X,M)$ and
  $(Y,N)$. Let $\mathcal{L}$ be the Boolean algebra generated by
  all the $\Alp^*$-languages of the form $L_1,L_2$ and $L_1aL_2$,
  where $L_1$ \textup{(}respectively $L_2$\textup{)} is recognised by
  $X$ \textup{(}respectively $Y$\textup{)} and $a\in \Alp$.  Then
  a $\Alp^*$-language is recognised by $X \Ds Y$ if, and only if,
  it belongs to $\mathcal{L}$.
\end{theorem}
\begin{proof}[Proof Idea]
Suppose the languages $L_1, L_2$ are recognised by morphisms $\phi_1\colon (\beta(\Alp^*),\Alp^*)\to (X,M)$ and 
$\phi_2\colon (\beta(\Alp^*),\Alp^*)\to (Y,N)$, respectively, and fix $a\in\Alp$.
 By abuse of notation, call $\phi_1\times \phi_2\colon \beta(\Alp^*\times \{a\}\times\Alp^*)\to X\times Y$ the unique continuous extension of the product map
$\Alp^*\times \{a\}\times\Alp^*\to X\times Y$ whose components are
$(w,a,w')\mapsto \phi_1(w)$ and $(w,a,w')\mapsto \phi_2(w')$.
Let $\zeta_a\colon \beta(\Alp^*)\to \V(X\times Y)$ be the continuous function induced by the diagram
\begin{equation}\label{eq:topo-span-zeta-a} \begin{tikzcd}[row sep=0.5em,column sep=3em]
 & \beta(\Alp^*\times \{a\}\times\Alp^*) \arrow{dl}[swap]{\beta c} \arrow{dr}{\phi_1\times \phi_2} & \\
 \beta(\Alp^*) & & X\times Y
\end{tikzcd} \end{equation}
just as for diagram~\eqref{eq:topo-span}, 
where $c\colon \Alp^*\times \{a\}\times\Alp^*\to \Alp^*$ is the concatenation map $(w,a,w')\mapsto waw'$.
One can prove that the map $\zeta_a$ is a morphism recognising $L_1,L_2$ and $L_1aL_2$.

Conversely, for any morphism $\langle\zeta,\phi_1,\phi_2\rangle\colon (\beta(\Alp^*),\Alp^*)\to(X \Ds Y,M \Ds N)$ and clopens $C_1\subseteq X$, $C_2\subseteq Y$, 
we must prove that $\zeta^{-1}(\Dv(C_1\times C_2))\cap\Alp^*\in\mathcal{L}$. One observes that each
\begin{align*}
L_{C_1\times C_2,a}:=\{w\in \Alp^*\mid \exists u,v\in \Alp^* \ \text{s.t.} \ w=uav \ \text{and} \ \phi_1(u)\zeta(a)\phi_2(v)\in \Dv (C_1\times C_2)\}
\end{align*}
is in the Boolean algebra $\mathcal{L}$. Then $\zeta^{-1}(\Dv(C_1\times C_2))\cap\Alp^*=\bigcup_{a\in\Alp} L_{C_1\times C_2,a}$.
\end{proof}
The next corollary follows at once by Theorem \ref{th:reutenauer-global}, by noting that $L_1L_2=\bigcup_{a\in\Alp}L_1a(a^{-1}L_2)$.
\begin{corollary}
\label{cor:schutz-for-spaces}
The Boolean space with an internal monoid $(\Ds(X,Y),\Ds(M,N))$ recognises the concatention $L_1L_2$ 
of languages $L_1$, $L_2$ recognised by $(X,M)$ and $(Y,N)$, respectively.
\end{corollary}
Finally, the following local statement is a direct consequence of the proof of Theorem \ref{th:reutenauer-global}.
\begin{theorem}[Reutenauer's theorem, local version]
\label{th:reutenauer-local}
  Consider morphisms $\phi_1\colon (\beta(\Alp^*),\Alp^*)\to
  (X,M)$ and $\phi_2\colon (\beta(\Alp^*),\Alp^*)\to (Y,N)$.
  Let $\mathcal{L}$ be the Boolean algebra generated by all the
  $\Alp^*$-languages of the form $L_1,L_2$ and $L_1aL_2$, where $L_1$
  \textup{(}respectively $L_2$\textup{)} is recognised by $\phi_1$
  \textup{(}respectively $\phi_2$\textup{)} and $a\in \Alp$. Then a
  $\Alp^*$-language is recognised by the morphism
  \begin{align*}
    \langle \langle \zeta_a\rangle_{a\in\Alp},\phi_1,\phi_2
    \rangle\colon\beta(\Alp^*)\to \V(X\times
    Y)^{\Alp}\times X\times Y 
  \end{align*}
where $\zeta_a\colon \beta(\Alp^*)\to \V(X\times Y)$ is induced by diagram~\eqref{eq:topo-span-zeta-a} if, and only if, it belongs to $\mathcal{L}$.
\end{theorem}

%
\section{Ultrafilter equations}\label{s:ultrafilter-equations}
Identifying simple equational bases for the Boolean algebras of languages recognised by Sch\"{u}tzenberger products, in terms of the equational theories of the input 
Boolean algebras, is an important step in studying classes built up by repeated application of quantification or language concatenation. See e.g. \cite{PW1996,BP2009} for examples of such work in the regular setting. 

As a proof-of-concept and first step, we provide a fairly easy to obtain completeness result for the Boolean algebra recognised by the local version of a Sch{\"u}tzenberger product of a space with the one element space. First we introduce notation for the dual construction, see Theorem~\ref{th:reutenauer-local}.
\begin{definition}
\label{def:schutz-BAs}
Let $\B_1$ and $\B_2$ be Boolean algebras of $\Alp^*$-languages closed under quotients. 
We define the \emph{binary Sch\"{u}tzenberger sum} of $\B_1$ and $\B_2$ to be the Boolean algebra of languages
\vspace*{-1cm}

\begin{align*}
\B_1\Dplus \B_2:=\langle \B_1\cup\B_2\cup\{L_1aL_2\mid L_1\in\B_1, \ L_2\in\B_2, \ a\in\Alp\} \rangle.
\end{align*}
\vspace*{-.5cm}

\noindent Note that this Boolean algebra is also closed under quotients.
\end{definition}
Let $\B\subseteq \P(\Alp^*)$ be a Boolean algebra closed under quotients. We give equations for $\B\Dplus 2$. Recall that an equation for a Boolean subalgebra of 
$\P(\Alp^*)$ is a pair $\mu\approx\nu$, where $\mu,\nu\in\beta(\Alp^*)$, and that $L\in\P(\Alp^*)$ \emph{satisfies the ultrafilter equation} $\mu\approx\nu$ provided 
\vspace*{-.7cm}

\begin{align*}
L\in \mu \quad \text{if, and only if,} \quad L\in\nu.
\end{align*}
\vspace*{-.7cm}

\noindent A Boolean subalgebra of $\P(\Alp^*)$ satisfies an ultrafilter equation provided each of its elements satisfies it. For background and more details on equations 
see e.g. \cite{GGP2008,GKP2016,Gehrke2016}. 
Now, set

\vspace*{-.7cm}
\begin{align*}
f_a\colon \Alp^*\otimes\nbb \to \Alp^*, \ (w,i)\mapsto w(a\text{@}i) \quad \text{and} \quad f_r\colon \Alp^*\otimes\nbb \to \Alp^*,\ (w,i)\mapsto w_{|i}=w_0\cdots w_{i-1}
\end{align*}

\vspace*{-.3cm}
\noindent where $a\in\Alp$ and $w(a\text{@}i)$ denotes the word obtained by replacing the $i$th letter of the word $w=w_0\cdots w_{|w|-1}$ by an $a$.

The intuition is that the extension $\beta f_a$ will allow us to \emph{factor} an ultrafilter at an occurrence of the letter $a$, whereas the extension $\beta f_r$ gives us
 access to the prefix of this factorisation.

\begin{definition}\label{d:eq} Let $\mathcal E(\B\Dplus 2)$ denote the set of all equations $\mu\approx\nu$ so that
\begin{itemize}
\item $\mu\approx\nu$ holds in $\B$;
\item for each $\gamma\in\beta(\Alp^*\otimes\nbb)$ so that $\mu=\beta f_a(\gamma)$, there exists $\delta\in\beta(\Alp^*\otimes\nbb)$ such that $\nu=\beta f_a(\delta)$ and the equation $\beta f_r(\gamma)\approx\beta f_r(\delta)$ holds in $\B$;
\item for each $\delta\in\beta(\Alp^*\otimes\nbb)$ so that $\nu=\beta f_a(\delta)$, there exists $\gamma\in\beta(\Alp^*\otimes\nbb)$ such that $\mu=\beta f_a(\gamma)$ and the equation $\beta f_r(\gamma)\approx\beta f_r(\delta)$ holds in $\B$.
\end{itemize}
\end{definition}
\begin{theorem}\label{t:ultrafilter-equations-completeness}
The ultrafilter equations in $\mathcal E(\B\emph{\Dplus} 2)$ characterise the Boolean algebra $\B\emph{\Dplus} 2$.
\end{theorem}
The proof of Theorem~\ref{t:ultrafilter-equations-completeness} relies on the following two lemmas.
\begin{lemma}\label{l:ultrafiter-image}
Let $\gamma\in\beta(\Alp^*\otimes\nbb)$. If $\mu=\beta f_a(\gamma)$ and $L\in \beta f_r(\gamma)$, then $La\Alp^*\in \mu$.
\end{lemma}
\begin{lemma}\label{l:ultrafiter-image-quasi-inverse}
Let $\mathscr{F}\subseteq \P(\Alp^*)$ be a proper filter, $\mu\in\beta(\Alp^*)$ and $a\in\Alp$. If $La\Alp^*\in\mu$ for all $L\in\mathscr{F}$, then there exists $\gamma\in\beta(\Alp^*\otimes \nbb)$ such that $\mu=\beta f_a(\gamma)$ and $\mathscr{F}\subseteq \beta f_r(\gamma)$.
\end{lemma}
\begin{proof}[Proof Idea for Theorem \ref{t:ultrafilter-equations-completeness}]
Soundness follows easily from the lemmas. For completeness notice that, by repeated use of compactness, $K\in\P(\Alp^*)$ belongs to $\B\Dplus 2$ if and only if for each $\mu\in\widehat{K}$, the clopen $\widehat{K}$ extends the set
\vspace*{-.6cm}

\begin{align*}
C_\mu:=\bigcap\{\widehat{L}\mid L\in\B,\ L\in\mu\}&\cap \bigcap\{\widehat{La\Alp^*}\mid a\in\Alp, L\in\B, La\Alp^*\in\mu\}\\
&\cap \bigcap\{(\widehat{La\Alp^*})^c \mid a\in\Alp, L\in\B, La\Alp^*\notin\mu\}.
\end{align*}
\vspace*{-.6cm}

\noindent Finally one shows, again using the lemmas, that $\mu\approx\nu\in\mathcal E(\B\Dplus 2)$ for any $\nu\in C_\mu$.
\end{proof}

\section{Conclusion}

In \cite{GGP2008} the concepts of recognition and of syntactic monoid,
stemming from the algebraic theory of regular languages, were seen to
naturally arise in the setting of Stone/Priestley duality for Boolean
algebras and lattices with additional operations. Reasoning by
analogy this lead in \cite{GGP2010} to the formulation of
generalisations, for arbitrary languages of finite words, of
recognition and syntactic objects in the setting of monoids equipped
with uniform space structures (so called \emph{semiuniform monoids}).
In this paper we naturally arrive at an isomorphic notion of
recogniser --- Boolean spaces with internal monoids --- which 
is however more amenable to existing tools from duality theory.


Our first contribution is setting up the right framework that allows us to
extend to the non-regular setting algebraic constructions whose
logical counterpart is adding a layer of quantifier depth.
We should mention that both the Sch{\"u}tzenberger and the block
product are algebraic constructions that can be used for this purpose
in the regular case. However, for technical reasons, extending the
former to Boolean spaces with internal monoids is more natural.
The unary Sch{\"u}tzenberger product that we introduce
 (which actually does not appear in the (pro)finite monoid
literature to the best of our knowledge) arises naturally via duality
for the Boolean algebra with quotients generated by the languages
$L_\exists$, for $L$ coming from some Boolean algebra $\B$. For lack
of space, we have not included this fairly involved dual computation
but have opted for introducing our product by analogy with the well-known 
one of Sch{\"u}tzenberger. 
Moreover, our framework can be easily extended to the case of bounded
distributive lattices, one would just need to use instead the Vietoris
functor on spectral spaces.

Furthermore, Theorem~\ref{th:recognised-by-diamond-X} of
Section~\ref{ss:unary-schutz} and Theorem~\ref{th:reutenauer-global}
of Section~\ref{s:binary-schutz-product}, provide
characterisations of the languages accepted by our unary and binary 
Sch{\"u}tzenberger products of Boolean spaces.
%
Finally, in Section~\ref{s:ultrafilter-equations} we derive a
preliminary result on equations. 
Theorem~\ref{t:ultrafilter-equations-completeness} on equational
completeness is by no means the final word, but rather a first stepping
stone in this direction. In the regular setting, as well as in the
special cases treated in \cite{GKP2016} and \cite{CK2016}, much smaller
subsets of $\mathcal E(\B\Dplus 2)$ have been shown to provide
complete axiomatisations. We expect that a notion akin to the derived
categories of profinite monoid theory \cite{Tilson1987} have to be
developed, and we expect the remainder of the Stone-\v Cech
compactification to play a key r\^ole in this.





\newpage

\appendix
\section{Addenda to Section \ref{s:recognition-spaces-dense-monoids}}\label{a:appendix-recognition-spaces}
We first provide more details regarding the connection between the
notion of Boolean space with an internal monoid
(Definition~\ref{d:spaces-with-internal-monoids}) and that of
semiuniform monoid~\cite{GGP2010}, as outlined in the Remark on
page~\pageref{remark-pervin}.
\begin{remark}
 As it was shown 
in \cite[Theorem~1.6]{GGP2010}, if $(M,\mathcal{U})$ is a semiuniform 
monoid, then its uniform completion $X$ is a Boolean space containing $M$ 
as a dense subspace. Also, by uniform continuity, the biaction of $M$ on 
itself has a unique extension to a biaction with continuous components on 
$X$. Thus $(X,M)$ is a Boolean space with an internal monoid. 

Conversely, given a Boolean space with an internal monoid $(X,M)$, since
preimages of clopens under the components of the actions of $M$ on X are
clopens, the actions of $M$ on itself are uniformly continuous with respect
to the Pervin uniformity $\mathcal{U}$ on $M$ given by the Boolean algebra 
$\B=\{C\cap M\mid C\ \text{is clopen in}\ X\}$. Thus $(M,\mathcal{U})$ is a 
semiuniform monoid. It is not hard to see that these two constructions are
inverse to each other.

\end{remark}

\begin{proof}[Proof of Lemma \ref{l:morphism-are-action-preserving}]
We shall only prove 
\begin{align}\label{eq:left-action-is-preserved}
f\circ \lambda_{m}=\lambda_{f(m)}\circ f
\end{align}
for all $m\in M$, since the proof for the right action is the same, mutatis mutandis.
For arbitrary elements $m,m'\in M$, note that
\begin{align*}
(f\circ \lambda_{m})(m')&=f(m\cdot m') \\
&=f(m)\cdot f(m') \\
&=(\lambda_{f(m)}\circ f)(m').
\end{align*}
In other words $f\circ \lambda_{m}$ and $\lambda_{f(m)}\circ f$ coincide on $M$.
It is well-known that, if two continuous maps into a Hausdorff space coincide on a dense subspace of the domain, then they are equal. 
Hence, $M$ being dense in $X$,~\eqref{eq:left-action-is-preserved} is proved.
\end{proof}

%
%

%
%

%
%
%

%
\section{Addenda to Section \ref{s:unary-schutz-product}}\label{a:appendix-unary-schutz-product}

\subsection{The Vietoris construction}\label{ss:Vietoris}
For any topological space $X$, denote by $\V(X)$ the collection of all closed subsets of $X$. Further, given $V\subseteq X$, set
\begin{align*}
\Dv V:=\{K\in\V(X)\mid K\cap V\neq \emptyset\}, \quad \text{and} \quad \boxa V:=\{K\in\V(X)\mid K\subseteq V\}.
\end{align*}
The set $\V(X)$, equipped with the topology\footnote{This is known in the literature as the \emph{exponential}, or \emph{finite}, topology on the space of closed subsets
of $X$.} having 
\begin{align*}
\{\Dv V\mid V\subseteq X \ \text{is open}\}\cup\{\boxa V\mid V\subseteq X \ \text{is open}\}
\end{align*}
as a subbasis of open sets, is called the \emph{Vietoris space} of $X$. 
Since the operator $\boxa$ preserves intersections (while $\Dv$ preserves unions), a basic open set for the latter topology is of the form 
$(\bigcap_{i=1}^{n-1}\Dv V_i)\cap \boxa V_n$, where $V_1,\ldots,V_n$ are open subsets of $X$.
%
%
%
We further note that, for any subset $V\subseteq X$, $\boxa V=(\Dv V^c)^c$.

The Vietoris construction preserves several topological properties of the space $X$ (the interested
reader is referred to \cite[\S 4]{Michael1951} for a complete account). The following preservation result is central in our treatment.
\begin{theorem}[{\cite[Theorem 4.9 p.\ 163]{Michael1951}}]
If $X$ is a Boolean space, then so is $\V(X)$. In this case, the topology of $\V(X)$ admits as a subbasis of clopen sets the collection
\begin{align*}
\{\Dv V\mid V\subseteq X \ \text{is clopen}\}\cup\{\boxa V\mid V\subseteq X \ \text{is clopen}\}.
\end{align*}
\end{theorem}
Henceforth, we shall assume that $X,Y$ are Boolean spaces. However, we remark that all the following facts hold in more generality.
Firstly, observe that the map 
\begin{align}\label{eq:embedding-of-X-into-vietoris}
\eta\colon X\to \V(X), \ x\mapsto\{x\}
\end{align}
is a continuous embedding of $X$ into its Vietoris space.
Secondly, if $f\colon X\to Y$ is a continuous map then the forward image function
\begin{align}\label{eq:vietoris-functor-on-maps}
\V(f)\colon \V(X)\to\V(Y), \ K\mapsto f[K]
\end{align}
is also continuous \cite[Theorem 5 p.\ 163]{Kuratowski1}.
Lastly, the following lemma shows that the Vietoris construction may be regarded as a generalisation of the finite power set.
\begin{lemma}[{\cite[Theorem 4 p.\ 163]{Kuratowski1}}]\label{l:finite-powerset-dense-in-vietoris}
If $X$ is a Boolean space, then $\Pfin(X)$ is dense in $X$. Therefore, if $Z$ is a dense subspace of $X$, then $\Pfin(Z)$ is dense in $X$.
\end{lemma}

\subsection{Proofs for Section \ref{s:unary-schutz-product}}
\begin{proof}[Proof of Proposition \ref{prop:monoid-blprod}]
Define the map $\xi\colon\Alp^*\to \Ds M$ as the pairing of
  $\xi_1\colon\Alp^*\to\Pfin(M)$ from~\eqref{eq:rel-PM}, and $\tau\circ
  \gamma_0\colon\Alp^*\to M$. Explicitly,
\[
w\mapsto (\{\tau(w^{(i)})\mid 0\le i< |w|\},\tau(w^0)).
\]
The latter is a monoid morphism since, for all $v,w\in\Alp^*$, 
\begin{align*}
\xi(v)*\xi(w)&=(\{\tau(v^{(i)})\ |\ 0\leq i< |v|\},\tau(v^0))*(\{\tau(w^{(i)})\ |\ 0\leq i< |w|\},\tau(w^0))\\
&=(\{\tau(v^{(i)})\ |\ 0\leq i< |v|\}\cdot\tau(w^0)\cup\tau(v^0)\cdot \{\tau(w^{(i)})\ |\ 0\leq i< |w|\},\tau(v^0)\cdot\tau(w^0))\\
&=(\{\tau(v^{(i)})\cdot\tau(w^0)\ |\ 0\leq i< |v|\}\cup \{\tau(v^0)\cdot\tau(w^{(i)})\ |\ 0\leq i< |w|\},\tau(v^0w^0))\\
&=(\{\tau((vw)^{(i)})\ |\ 0\leq i< |v|\}\cup \{\tau((vw)^{(i+|v|)})\ |\ 0\leq i< |w|\},\tau((vw)^0))\\
&=(\{\tau((vw)^{(i)})\ |\ 0\leq i< |v|+|w|\},\tau((vw)^0))=\xi(vw).
\end{align*}
In order to see that $\xi$ recognises the language $L_{\exists x.\Phi}$, consider a subset $V\subseteq M$ such that $L_\Phi=\tau^{-1}(V)$, and
set $\Dv V:=\{S\in\Pfin(M)\mid S\cap V\neq\emptyset\}$. Then 
\begin{align*}
\xi^{-1}(\Dv V\times M)&=\{w\in \Alp^*\mid \{\tau(w^{(i)})\mid 0\leq i< |w|\}\in\Dv V\} \\
&=\{w\in \Alp^*\mid \{\tau(w^{(i)})\mid 0\leq i< |w|\}\cap V\neq\emptyset\} \\
&=\{w\in \Alp^*\mid \exists 0\leq i< |w| \ \text{s.t.} \ w^{(i)}\in \tau^{-1}(V)\}=L_{\exists x.\Phi}.\qedhere
\end{align*}
\end{proof}
%

%
%
%

%
\begin{proof}[Proof of Lemma \ref{l:unary-schutz-product-is-a-bin}]
In view of Lemma~\ref{l:finite-powerset-dense-in-vietoris}, $\Pfin(M)$ is a dense subspace of $\V(X)$. 
Thus the monoid $\Ds M$ is a dense subspace of $\Ds X$.
We show that, for each $S\in\Pfin(M)$ and $m\in M$, the function $l_{(S,m)}\colon\Ds X\to \Ds X$ given by
\begin{align*}
l_{(S,m)}(K,x):=(\{\lambda_s(x)\mid s\in S\}\cup\lambda_m[K],\lambda_m(x))
\end{align*}
is continuous. It is clear that the above map extends the left action of $\Ds M$ on itself. Uniqueness will then follow automatically from continuity.
The continuity of the right action can be proved in a similar fashion.

Note that it suffices to prove that $(\pi_1\circ l_{(S,m)})^{-1}(\Dv V)$ is clopen whenever $V\subseteq X$ is clopen, 
where $\pi_1\colon \Ds X\to \V(X)$ is the first projection. Then
\begin{align*}
(\pi_1\circ l_{(S,m)})^{-1}(\Dv V)&=\{(K,x)\in\V(X)\times X\mid (\{\lambda_s(x)\mid s\in S\}\cup \lambda_m[K])\cap V\neq\emptyset\} \\
&=\{(K,x)\in\V(X)\times X\mid \exists s\in S \ \text{s.t.} \ \lambda_s(x)\in V\}\cup (\Dv \lambda_m^{-1}(V)\times X) \\
&=(\V(X)\times \bigcup_{s\in S}\lambda_{s}^{-1}(V))\cup(\Dv \lambda_m^{-1}(V)\times X)
\end{align*}
showing $(\pi_1\circ l_{(S,m)})^{-1}(\Dv V)$ as a clopen in $\Ds X$.
\end{proof}
\begin{proof}[Proof of Proposition \ref{prop:recog-ex-unary-schutz}]
The map $\xi\colon\beta(\Alp^*)\to\Ds X$ can be defined as the pairing of the map $\xi_1\colon\beta(\Alp^*)\to\V(X)$ from~\eqref{eq:map-xi-1-for-spaces} with 
$\tau\circ\beta\gamma_0$. This is clearly continuous, and it restricts to a monoid morphism $\Alp^*\to\Ds M$ by (the proof of) Proposition \ref{prop:monoid-blprod}.

If the morphism $\tau$ recognises the language $L_{\Phi}$ through the clopen $V\subseteq X$, it is easy to see that $\xi$ recognises the language $L_{\exists x.\Phi}$
 through the clopen $\Dv V\times X$.
\end{proof}
\begin{proof}[Proof of Theorem \ref{th:recognised-by-diamond-X}]
\emph{Right-to-left:} pick a language $L\in \B(X,\Alp)$. Then there is a clopen $V\subseteq X$ and a morphism 
$f\colon(\beta(\Alp^+),\Alp^+)\to (X,M)$ satisfying
$\wL=f^{-1}(V)$. Define $g\colon\beta(\Alp^+)\to \Ds X$ as the composition
\begin{equation*}
\begin{tikzcd}
\beta(\Alp^+) \arrow{r}{\langle f,f\rangle} & X\times X \arrow{r}{\eta\times id_X} & \V(X)\times X
\end{tikzcd}
\end{equation*}
where $\eta\colon X\to\V(X)$ is the canonical embedding from~\eqref{eq:embedding-of-X-into-vietoris}. 
Since clearly $g^{-1}(\V(X)\times V)=\wL$, it is enough to show that $g$ restricts
to a semigroup morphism $\Alp^+\to \Ds M$.
For each $w,w'\in\Alp^+$
\begin{align*}
g(w)\cdot g(w')&=(\{f(w)\},f(w))\ast (\{f(w')\},f(w')) \\
&=(\{f(w)\}\cdot f(w') \cup f(w)\cdot\{f(w')\}, f(ww')) \\
&=(\{f(ww')\},f(ww'))=g(ww').
\end{align*}

On the other hand, if $L\in \B(X,\Atw)$ we have $\wL=f^{-1}(V)$ for some morphism $f\colon \beta(\Atw)^+\to X$ and some clopen $V\subseteq X$. 
Consider the clopen subset $\Dv V$ of $\V(X)$.
We claim that the map $\zeta:=\langle f\circ\beta\gamma_1\circ(\beta\pi)^{-1},f\circ\beta\gamma_0\rangle\colon \beta(\Alp^+)\to \Ds X$ 
recognises $L_{\exists}$ through the clopen $\Dv V\times X$.

  In fact it suffices to show that
  $(f\circ\beta\gamma_1\circ(\beta\pi)^{-1})^{-1}(\Dv V)=\widehat{L_{\exists}}$, where we recall that $L_{\exists}:=\pi(\gamma_1^{-1}(L))$. 
  This is done in the following computation.
  \begin{align*}
    (f\circ\beta\gamma_1\circ(\beta\pi)^{-1})^{-1}(\Dv V) \cap \Alp^* &=\{w\in \Alp^+\mid (f\circ\beta\gamma_1\circ(\beta\pi)^{-1})(\uparrow w)\cap V\neq\emptyset\} \\
    &=\{w\in \Alp^+\mid (\beta\gamma_1\circ(\beta\pi)^{-1})(\uparrow w)\cap \wL\neq\emptyset\}  \\
    &=\{w\in \Alp^+\mid \uparrow w\in \widehat{L_{\exists}}\}=L_{\exists}.
  \end{align*}
The fact that $\zeta$ restricts to a semigroup morphism follows at once from the monoid case (see Proposition \ref{prop:monoid-blprod}). 

\medskip
\emph{Left-to-right:} it is enough to prove the statement for every $L\in \B(\Ds X,\Alp)$ satisfying $\wL=f^{-1}(\Dv V\times C)$,
where $f\colon (\beta(\Alp^+),\Alp^+)\to (\Ds X,\Ds M)$ is a morphism and $V,C$ are clopens of $X$. 
If $f=\langle \sigma,h\rangle$, then
\begin{align*} 
f^{-1}(\Dv V\times C)=\sigma^{-1}(\Dv V)\cap h^{-1}(C). 
\end{align*}
Since the projection on the second component $\Ds X\to X$ is a morphism, $h^{-1}(C)\in \B(X,\Alp)$.
We will prove $\sigma^{-1}(\Dv V)\in \B(X,\Atw)_\exists$, and this will complete the proof.

Note that $\sigma$ restricts to a map $\Alp^+\to \Pfin^{+}(M)$, hence we can define a finite non-empty set $I:=\prod_{a\in\Alp} \sigma(a)$.
Each $m=(m_a)_{a\in\Alp}\in I$ defines a semigroup morphism $\tau_m\colon (\Atw)^+\to M$ whose behaviour on the generators is given by
\begin{align*}
\tau_m(a,0):=h(a), \quad \tau_m(a,1):=m_a.
\end{align*}
By the universal property~\eqref{eq:stone-cech-universal-property} of the Stone-\v{C}ech compactification, the maps $\tau_m$ can be uniquely extended to continuous functions $\beta(\Atw)^+\to X$ that
we denote again by $\tau_m$. It is clear that the latter maps are morphisms $(\beta(\Atw)^+,(\Atw)^+)\to (X,M)$. We claim that
\begin{align}\label{eq:sem-claim}
\sigma^{-1}(\Dv V)=\bigcup_{m\in I} (\tau_m^{-1}(V))_{\exists}.
\end{align}
Since each $\tau_m^{-1}(V)$ belongs to $\B(X,\Atw)$, this will exhibit $\sigma^{-1}(\Dv V)$ as a finite union of elements of $\B(X,\Atw)_{\exists}$.

Now, by a straightforward translation of a fact noticed in \cite[p.\ 261]{Reutenauer1979}, for any $w\in\Alp^+$
\begin{align*}
\sigma(w)=\bigcup_{\substack{u,v\in \Alp^+ \\ a\in \Alp \\ w=uav}}h(u)\sigma(a)h(v) \cup \bigcup_{\substack{u\in \Alp^+ \\ a\in \Alp \\ w=ua}}h(u)\sigma(a)\cup
\bigcup_{\substack{v\in \Alp^+ \\ a\in \Alp \\ w=av}}\sigma(a)h(v) \cup \bigcup_{\substack{a\in \Alp \\ w=a}}\sigma(a).
\end{align*}
Thus
\begin{align*}
\sigma^{-1}(\Dv V)= \ &\{w\in\Alp^+\mid \exists a\in\Alp, \exists u,v\in\Alp^+ \ \text{s.t.} \ w=uav, \exists m_a\in \sigma(a) \ \text{s.t.} \ h(u)m_a h(v)\in V\} \\
&\cup \{w\in\Alp^+\mid \exists a\in\Alp, \exists u\in\Alp^+ \ \text{s.t.} \ w=ua, \exists m_a\in \sigma(a) \ \text{s.t.} \ h(u)m_a \in V\} \\
&\cup \{w\in\Alp^+\mid \exists a\in\Alp, \exists v\in\Alp^+ \ \text{s.t.} \ w=av, \exists m_a\in \sigma(a) \ \text{s.t.} \ m_a h(v) \in V\} \\
&\cup \{w\in\Alp^+\mid \exists a\in\Alp \ \text{s.t.} \ w=a, \exists m_a\in \sigma(a) \ \text{s.t.} \ m_a \in V\} \\
= \ &\{w\in \Alp^+\mid \exists m\in I, \exists 0\leq n<|w| \ \text{s.t.} \ \tau_m(w^{(n)})\in V\}=\bigcup_{m\in I} (\tau_m^{-1}(V))_{\exists}
\end{align*}
and~\eqref{eq:sem-claim} is proved.
\end{proof}

\section{Addenda to Section \ref{s:binary-schutz-product}}\label{a:appendix-binary-schutz-product}
\begin{proof}[Proof of Lemma \ref{l:action-on-binary-schutz-continuous}]
We show that the components of the left action are continuous, the proof for the right action being the same, mutatis mutandis.
It suffices to prove that the map
\begin{align*}
g\colon \V(X\times Y)\times Y \to \V(X\times Y), \ (Z,y)\mapsto m_1 Z\cup S y
\end{align*}
is continuous, for every $m_1\in M$ and $S\in\Pfin(M\times N)$. Let $L_1, L_2$ be clopens in $X$ and $Y$, respectively. Then
\begin{align*}
g^{-1}(\boxa (L_1\times L_2))&=\{(Z,y)\in \V(X\times Y)\times Y\mid m_1 Z\cup S y\subseteq L_1\times L_2\} \\
&=\{(Z,y)\in \V(X\times Y)\times Y\mid m_1 Z\subseteq L_1\times L_2, \ S y\subseteq L_1\times L_2\}.
\end{align*}
Observe that
\begin{align*}
m_1 Z=\{(\lambda_{m_1}(x),y)\in X\times Y\mid (x,y)\in Z\}\subseteq L_1\times L_2 &\Longleftrightarrow \\
y\in \lambda_{m_1}^{-1}(L_1) \ \text{and} \ y\in L_2, \ \forall (x,y)\in Z &\Longleftrightarrow \\
Z\subseteq  \lambda_{m_1}^{-1}(L_1)\times L_2.
\end{align*}
Similarly,
\begin{align*}
S y=\{(m,\lambda_{n}(y))\in X\times Y\mid (m,n)\in S\}\subseteq L_1\times L_2 &\Longleftrightarrow \\
m\in L_1 \ \text{and} \ y\in \lambda_{n}^{-1}(L_2), \, \, \forall (m,n)\in S &\Longleftrightarrow \\
\pi_1(S)\subseteq L_1 \ \text{and} \ y\in \bigcap_{n\in \pi_2(S)}\lambda_{n}^{-1}(L_2).
\end{align*}
If $\pi_1(S)\nsubseteq L_1$, then $g^{-1}(\boxa (L_1\times L_2))=\emptyset$. Otherwise
\begin{align*}
g^{-1}(\boxa (L_1\times L_2))&=\{(Z,y)\in \V(X\times Y)\times Y\mid Z\subseteq  \lambda_{m_1}^{-1}(L_1)\times L_2, \, y\in \bigcap_{n_2\in \pi_2(S)}
\lambda_{n_2}^{-1}(L_2)\} \\
&=\left(\bigcap_{n\in \pi_2(S)} \lambda_{n}^{-1}(L_2)\right)\times\left(\boxa (\lambda_{m_1}^{-1}(L_1)\times L_2) \right),
\end{align*}
exhibiting $g^{-1}(\boxa (L_1\times L_2))$ as a clopen. On the other hand,
\begin{align*}
g^{-1}(\Dv (L_1\times L_2))&=\{(Z,y)\in \V(X\times Y)\times Y\mid (m_1 Z\cup S y)\cap (L_1\times L_2)\neq\emptyset\} \\
&=(\V(X\times Y)\times \{y\mid S y\cap (L_1\times L_2)\neq\emptyset\})\cup (\{Z\mid m_1 Z\cap(L_1\times L_2)\neq \emptyset\}\times Y).
\end{align*}
We remark that
\begin{align*}
m_1 Z\cap (L_1\times L_2)\neq\emptyset &\Longleftrightarrow \\
\{(\lambda_{m_1}(x),y)\in X\times Y\mid (x,y)\in Z\} \cap (L_1\times L_2)\neq\emptyset &\Longleftrightarrow \\
\exists (x,y)\in Z \ \text{s.t.} \ x\in\lambda_{m_1}^{-1}(L_1) \ \text{and} \ y\in L_2 &\Longleftrightarrow \\
Z\in \Dv(\lambda_{m_1}^{-1}(L_1)\times L_2)
\end{align*}
and
\begin{align*}
S y\cap (L_1\times L_2)\neq\emptyset &\Longleftrightarrow \\
\{(m,\lambda_{n}(y))\in X\times Y\mid (m,n)\in S\} \cap (L_1\times L_2)\neq\emptyset &\Longleftrightarrow \\
\exists (m,n)\in S \ \text{s.t.} \ m\in L_1 \ \text{and} \ y\in \lambda_{n}^{-1}(L_2) &\Longleftrightarrow \\
\pi_1(S)\cap L_1\neq\emptyset \ \text{and} \ y\in \bigcup_{n\in \pi_2(T)}\lambda_{n}^{-1}(L_2),
\end{align*}
where $T:=\pi_1^{-1}(\pi_1(S)\cap L_1)$.
Therefore
\begin{align*}
g^{-1}(\Dv (L_1\times L_2))&=\left(\V(X\times Y)\times\left(\bigcup_{n\in \pi_2(T)}\lambda_{n}^{-1}(L_2)\right)\right)\cup (\Dv(\lambda_{m_1}^{-1}(L_1)\times L_2)\times Y),
\end{align*}
showing $g^{-1}(\Dv (L_1\times L_2))$ as a clopen, and this completes the proof.
\end{proof}

\begin{proof}[Proof of Theorem \ref{th:reutenauer-global}]
Suppose that the languages $L_1, L_2$ are recognised by morphisms $\phi_1\colon (\beta(\Alp^*),\Alp^*)\to (X,M)$ and 
$\phi_2\colon (\beta(\Alp^*),\Alp^*)\to (Y,N)$ through the clopens $C_1\subseteq X$ and $C_2\subseteq Y$,
 respectively. For an arbitrarily fixed $a\in\Alp$, we will define a morphism $(\beta(\Alp^*),\Alp^*)\to(X \Ds Y,M \Ds N)$ recognising the language $L_1aL_2$.

 By abuse of notation, we denote $\phi_1\times \phi_2\colon \beta(\Alp^*\times \{a\}\times\Alp^*)\to X\times Y$ the unique continuous extension of the product map
$\Alp^*\times \{a\}\times\Alp^*\to X\times Y$ whose components are
\begin{align*}
(w,a,w')\mapsto \phi_1(w),\quad \text{and} \quad (w,a,w')\mapsto \phi_2(w').
\end{align*}
Let $\zeta_a\colon \beta(\Alp^*)\to \V(X\times Y)$ be the continuous function induced by the diagram
\[ \begin{tikzcd}[row sep=0.5em,column sep=3em]
 & \beta(\Alp^*\times \{a\}\times\Alp^*) \arrow{dl}[swap]{\beta c} \arrow{dr}{\phi_1\times \phi_2} & \\
 \beta(\Alp^*) & & X\times Y
\end{tikzcd} \]
just as for diagram~\eqref{eq:topo-span}, 
where $c\colon \Alp^*\times \{a\}\times\Alp^*\to \Alp^*$ is the concatenation map $(w,a,w')\mapsto waw'$.
We claim that the map $\zeta_a$ recognises the language $L_1aL_2$ through the clopen $\Dv(C_1\times C_2)$. Indeed, 
%
\begin{align*}
\zeta_a^{-1}(\Dv(C_1\times C_2))\cap\Alp^*&=\{w\in \Alp^*\mid ((\phi_1\times\phi_2)\circ(\beta c)^{-1}(\uparrow w))\cap (C_1\times C_2)\neq\emptyset\} \\
&=\{w\in \Alp^*\mid (\beta c)^{-1}(\uparrow w)\cap (\phi_1\times\phi_2)^{-1}(C_1\times C_2)\neq\emptyset\} \\
&=\{w\in \Alp^*\mid (\beta c)^{-1}(\uparrow w)\cap \overline{(L_1\times\{a\}\times L_2)}\neq\emptyset\} \\
&=\{w\in \Alp^*\mid \exists u\in L_1, \ \exists v\in L_2 \ \text{s.t.} \ w=uav\}=L_1 a L_2.
\end{align*}
 Therefore the continuous product map $\langle\zeta_a,\phi_1,\phi_2\rangle\colon\beta(\Alp^*)\to X \Ds Y$ recognises 
 the language $L_1aL_2$ through the clopen $\Dv(C_1\times C_2)\times X\times Y$. Moreover, the latter map induces a morphism 
 $(\beta(\Alp^*),\Alp^*)\to(X \Ds Y,M \Ds N)$ because $\phi_1,\phi_2$ restrict to monoid morphisms, and for all $w,w'\in\Alp^*$
 \begin{align*}
 \phi_1(w)\cdot \zeta_a(w')\cup \zeta_a(w)\cdot\phi_2(w')=&\ \phi_1(w)\cdot\{(\phi_1(u),\phi_2(v))\mid u,v\in\Alp^*, w'=uav\}\ \cup \\
 &\ \{(\phi_1(u),\phi_2(v))\mid u,v\in\Alp^*, w=uav\}\cdot\phi_2(w') \\
 =&\ \{(\phi_1(wu),\phi_2(v))\mid u,v\in\Alp^*, w'=uav\}\ \cup \\
 &\ \{(\phi_1(u),\phi_2(vw'))\mid u,v\in\Alp^*, w=uav\} \\
 =&\ \zeta_a(ww').
 \end{align*}
 We remark that the morphism $\langle\zeta_a,\phi_1,\phi_2\rangle\colon (\beta(\Alp^*),\Alp^*)\to(X \Ds Y,M \Ds N)$ recognises also the languages $L_1$ and
 $L_2$ through the clopens $\V(X\times Y)\times C_1\times Y$ and $\V(X\times Y)\times X\times C_2$.
%

For the converse direction, consider an arbitrary morphism 
\begin{align*}
\langle\zeta,\phi_1,\phi_2\rangle\colon (\beta(\Alp^*),\Alp^*)\to(X \Ds Y,M \Ds N).
\end{align*}
It suffices to show that the language $\zeta^{-1}(\Dv(C_1\times C_2))\cap\Alp^*$ belongs to the Boolean algebra $\mathcal{L}$, for arbitrary clopens $C_1\subseteq X$ 
and $C_2\subseteq Y$.
We shall need the following
\begin{claim*}\label{claim:languages-C1-C2-a-are-in-L}
If $a\in\Alp$ and $C_1$ and $C_2$ are clopens of $X$ and $Y$, respectively, then
\begin{align*}
L_{C_1\times C_2,a}:=\{w\in \Alp^*\mid \exists u,v\in \Alp^* \ \text{s.t.} \ w=uav \ \text{and} \ \phi_1(u)\zeta(a)\phi_2(v)\in \Dv (C_1\times C_2)\}
\end{align*}
belongs to the Boolean algebra $\mathcal{L}$.
\end{claim*}
\begin{proof}[Proof of Claim]
Since $\zeta(a)\in \Pfin(M\times N)$, there is $s\in\nbb$ such that
\begin{align*}
\zeta(a)=\{(m_1,n_1),\ldots,(m_s,n_s)\}
\end{align*}
for some $\{m_i\}_{i=1}^{s}\subseteq M$ and $\{n_i\}_{i=1}^{s}\subseteq N$. We show that
\begin{align}\label{eq:finite-union-in-L}
L_{C_1\times C_2,a}=\bigcup_{i=1}^s A_i a B_i
\end{align}
where $A_i:=\phi_1^{-1}(\rho_{m_i}^{-1}(C_1))\cap\Alp^*$ and $B_i:=\phi_2^{-1}(\lambda_{n_i}^{-1}(C_2))\cap\Alp^*$ 
(recall that $\rho_{m_i}$ is the continuous component of the right action 
of $M$ on $X$, and $\lambda_{n_i}$ is the continuous component of the left action of $N$ on $Y$). This will settle the claim.

Pick $w\in \Alp^*$. Then $w\in L_{C_1\times C_2,a}$ if, and only if, there exist $u,v\in \Alp^*$ with $w=uav$ and $\phi_1(u)\zeta(a)\phi_2(v)\in\Dv (C_1\times C_2)$ if,
and only if, $w=uav$ and there is $i\in \{1,\ldots, s\}$ such that 
\begin{align*}
(\phi_1(u)\cdot m_i,n_i\cdot \phi_2(v))=\phi_1(u)\cdot (m_i,n_i)\cdot \phi_2(v)\in C_1\times C_2,
\end{align*}
i.e.\ $u\in \phi_1^{-1}(\rho_{m_i}^{-1}(C_1))\cap\Alp^*$ and $v\in \phi_2^{-1}(\lambda_{n_i}^{-1}(C_2))\cap\Alp^*$. 
In turn, this is equivalent to $w\in\bigcup_{i=1}^s A_i a B_i$ and (\ref{eq:finite-union-in-L}) is proved.
\end{proof}
Now, as observed in \cite[p.\ 261]{Reutenauer1979}, for any $w\in \Alp^*$
\begin{align*}
\zeta(w)=\bigcup_{\substack{u,v\in \Alp^* \\ a\in \Alp \\ w=uav}}\phi_1(u)\zeta(a)\phi_2(v).
\end{align*}
Thus $w\in \zeta^{-1}(\Dv(C_1\times C_2))\cap\Alp^*$ if, and only if, there are $u,v\in\Alp^*$ and $a\in \Alp$ such that 
$w=uav$ and $\phi_1(u)\zeta(a)\phi_2(v)\in\Dv (C_1\times C_2)$. Therefore
\begin{align*}
\zeta^{-1}(\Dv(C_1\times C_2))\cap\Alp^*=\bigcup_{a\in\Alp} L_{C_1\times C_2,a}
\end{align*}
which, by the claim, exhibits $\zeta^{-1}(\Dv(C_1\times C_2))\cap\Alp^*$ as a finite union of elements of 
$\mathcal{L}$.
\end{proof}
\section{Addenda to Section \ref{s:ultrafilter-equations}}\label{a:appendix-ultrafilter-equations}
We recall that a subset $\mathscr{S}$ of a Boolean algebra $(\B,\wedge,\vee,\neg,0,1)$ is a \emph{filter base} if it has the finite intersection property, 
that is $L_1\wedge\cdots\wedge L_n\neq 0$ for any $L_1,\ldots,L_n\in \mathscr{S}$.
%
%
%
%
\begin{proof}[Proof of Theorem \ref{t:ultrafilter-equations-completeness}]
We first prove \emph{soundness}, i.e.\ every element of $\B\Dplus 2$ satisfies the set of ultrafilter equations $\mathcal E(\B\Dplus 2)$.
It is enough to check that, for any $L\in\B$, $a\in\Alp$ and $\mu\approx \nu\in \mathcal E(\B\Dplus 2)$, the language $La\Alp^*$ belongs to $\nu$ whenever 
it belongs to $\mu$.
By applying Lemma \ref{l:ultrafiter-image-quasi-inverse} with $\mathscr{F}:=\{L\}$, the condition $La\Alp^*\in\mu$ entails that there exists 
$\gamma\in\beta(\Alp^*\otimes \nbb)$ such that $\mu=\beta f_a(\gamma)$ and $L\in \beta f_r(\gamma)$.
Then, by hypothesis, there is $\delta\in\beta(\Alp^*\otimes\nbb)$ satisfying $\nu=\beta f_a(\delta)$ and $L\in\beta f_r(\delta)$.
Hence $La\Alp^*\in\nu$ by Lemma \ref{l:ultrafiter-image}.
%

Now, we prove \emph{completeness}: every language $K\in\P(\Alp^*)$ satisfying all the equations in $\mathcal E(\B\Dplus 2)$ must belong to $\B\Dplus 2$. 
Let us denote the dual map of the embedding $\B\hookrightarrow \P(\Alp^*)$ by $\phi\colon \beta(\Alp^*)\to X$, and for any ultrafilter $\mu\in\widehat{K}$ set
\begin{align*}
C_{\mu}:=\phi^{-1}(\phi(\mu))&\cap \bigcap\{\widehat{La\Alp^*}\mid a\in\Alp, L\in\B, La\Alp^*\in\mu\} \\
&\cap \bigcap\{(\widehat{La\Alp^*})^c \mid a\in\Alp, L\in\B, La\Alp^*\notin\mu\}.
\end{align*}
\begin{claim*}
Let $K\in\P(\Alp^*)$. Then $K\in\B\emph{\Dplus} 2$ if, and only if, $C_{\mu}\subseteq \widehat{K}$ for all $\mu\in\widehat{K}$.
\end{claim*}
\begin{proof}[Proof of Claim]
Let $\mu$ be an arbitrary element of $\widehat{K}$, and assume that $C_{\mu}\subseteq \widehat{K}$. Then
\begin{align*}
\phi^{-1}(\phi(\mu))=\bigcap\{\widehat{H}\mid H\in\B,\ H\in\mu\}.
\end{align*}
By compactness there are $H_1,\ldots,H_h,L_1,\ldots,L_l,M_1,\ldots,M_m\in\B$ such that
\begin{align*}
D_{\mu}:=(\bigcap_{i=1}^h\widehat{H_i})\cap(\bigcap_{i=1}^l \widehat{L_ia_i\Alp^*})\cap(\bigcap_{i=1}^m(\widehat{M_ia'_i\Alp^*})^c)\subseteq \widehat{K}.
\end{align*}
Then $D_{\mu}$ is a clopen containing $\mu$, and $L_{\mu}:=D_{\mu}\cap \Alp^*\in\B\Dplus 2$. Moreover $\widehat{L_{\mu}}=D_{\mu}\subseteq \widehat{K}$, hence
$\widehat{K}=\bigcup_{\mu\in\widehat{K}}\widehat{L_{\mu}}$ since $\mu$ is arbitrary. Again by compactness there are $\mu_1,\ldots,\mu_n\in\widehat{K}$ such that 
$\widehat{K}=\bigcup_{i=1}^n\widehat{L_{\mu_i}}$. Thus $K\in\B\Dplus 2$.

For the converse direction, pick $\nu\in C_{\mu}$, for some $\mu\in\widehat{K}$. 
Then $\B\Dplus 2$ satisfies the equation $\mu\approx\nu$. Since $K\in\B\Dplus 2$ and $\mu\in\widehat{K}$, we have $K\in\nu$, 
i.e.\ $\nu\in\widehat{K}$.
\end{proof}
In view of the previous claim it is enough to fix an arbitrary $\mu\in\widehat{K}$ and show that $C_{\mu}\subseteq \widehat{K}$. 
Pick $\nu\in C_{\mu}$ and notice that it suffices to prove 
$\mu\approx\nu\in \mathcal E(\B\Dplus 2)$, for then $\mu\in\widehat{K}$ entails $\nu\in\widehat{K}$, since $K$ is assumed to
satisfy all equations in $\mathcal E(\B\Dplus 2)$.

Clearly, $\nu\in\phi^{-1}(\phi(\mu))$ entails that $\mu\approx\nu$ holds in $\B$.
For the second condition in Definition \ref{d:eq}, suppose that $\mu=\beta f_a(\gamma)$ for some $\gamma\in\beta(\Alp^*\otimes \nbb)$, and consider the collection
\begin{align*}
\mathscr{S}:=\{L\mid L\in\B,\ L\in\beta f_r(\gamma)\}.
\end{align*}
Then $La\Alp^*\in\mu$ for every $L\in\mathscr{S}$, by Lemma \ref{l:ultrafiter-image}. Moreover, since $\mu\approx\nu$ holds in $\B$, $La\Alp^*\in\nu$ for all 
$L\in\mathscr{S}$. Since $\mathscr{S}$ is a filter base closed under finite intersections, (upon considering the proper filter generated by $\mathscr{S}$) Lemma 
\ref{l:ultrafiter-image-quasi-inverse} entails the existence of
$\delta\in\beta(\Alp^*\otimes\nbb)$ such that $\nu=\beta f_a(\delta)$ and $\mathscr{S}\subseteq \beta f_r(\delta)$.
Notice that $\mathscr{S}=\phi(\beta f_r(\gamma))$, thus $\phi(\beta f_r(\gamma))=\phi(\beta f_r(\delta))$, that is $\B$ satisfies the equation 
$\beta f_r(\gamma)\approx\beta f_r(\delta)$.

The third condition can be proved in a similar fashion.
\end{proof}
\begin{proof}[Proof of Lemma \ref{l:ultrafiter-image}]
Recall from~\eqref{eq:beta-on-maps} that the condition $L\in \beta f_r(\gamma)$ means $f_r^{-1}(L)\in \gamma$. Moreover
\begin{align*}
f_r^{-1}(L)=\{(w,i)\in\Alp^*\otimes\nbb\mid w_{|i}\in L\}\subseteq \{(w,i) \in\Alp^*\otimes\nbb\mid w(a\text{@}i)\in La\Alp^*\}=f_a^{-1}(La\Alp^*)
\end{align*}
so that $f_a^{-1}(La\Alp^*)\in \gamma$, i.e.\ $La\Alp^*\in \beta f_a(\gamma)=\mu$.
\end{proof}
\begin{proof}[Proof of Lemma \ref{l:ultrafiter-image-quasi-inverse}]
It suffices to show that the collection
\begin{align*}
\{f_a^{-1}(K)\cap f_r^{-1}(L)\mid K\in\mu,\ L\in\mathscr{F}\}
\end{align*}
is a filter base, for then any ultrafilter extending this base will satisfy the conditions in the statement. 
Furthermore, since $\mu$ and $\mathscr{F}$ are closed under finite 
intersections, it is enough to show that each set $f_a^{-1}(K)\cap f_r^{-1}(L)$ is not empty.
 
Since $La\Alp^*\in\mu$ by hypothesis, the intersection $K\cap La\Alp^*$ is non-empty because it belongs to $\mu$. Thus there exists $w\in K$ and $0\leq i<|w|$ such that
$w_{|i}\in L$ and $w_i=a$. That is, $(w,i)\in f_a^{-1}(K)\cap f_r^{-1}(L)$.
\end{proof}
\end{document}